\documentclass[conference,letterpaper]{IEEEtran}
\IEEEoverridecommandlockouts

\usepackage{cite}
\usepackage[tbtags]{amsmath}
\usepackage{amssymb}
\usepackage{amsfonts}
\usepackage{algorithmic}
\usepackage{graphicx}
\usepackage{textcomp}
\usepackage[dvipsnames]{xcolor}
\usepackage{siunitx}
\usepackage[printonlyused]{acronym}
\usepackage{MnSymbol}
\usepackage{tikz}
\usepackage{booktabs}
\usepackage{threeparttable}
\usetikzlibrary{quotes,angles,positioning}
\usetikzlibrary{decorations.pathreplacing}
\usepackage[english]{babel}
\usepackage{amsthm}
\usepackage{graphics}
\usepackage{siunitx}
\usepackage[caption=false,font=footnotesize]{subfig}

\DeclareMathOperator{\EX}{\mathbb{E}}% expected value
\newcommand{\var}{\mathrm{Var}} % variance 

\pgfdeclarelayer{bg} % declare background layer
\pgfsetlayers{bg,main}

\newcommand{\isep}{\mathrel{{.}\,{.}}\nobreak}

\title{Data Freshness in Mixed-Memory Intermittently-Powered Systems}

\author{\IEEEauthorblockN{James Scott Broadhead and Przemys\l{}aw Pawe\l{}czak}
\IEEEauthorblockA{Embedded and Networked Systems, EEMCS, Delft University of Technology, The Netherlands\\
Email: \{J.S.Broadhead, P.Pawelczak\}@tudelft.nl}}

\newtheorem{theorem}{Theorem}
\newtheorem{lemma}[theorem]{Lemma}

\begin{document}

\maketitle

\begin{abstract}
\ac{aoi} is a key metric to understand data freshness in \ac{iot} devices. In this paper we analyse an intermittently-powered \ac{iot} sensor---with mixed-memory (volatile and non-volatile) architecture---that uses a \ac{tdc} scheme. We derive the average \ac{paoi} and average \ac{aoi} of the system, and use these metrics to understand which device parameters most significantly influence performance. We go on to consider how the average \ac{paoi} of a mixed-memory system compares with entirely volatile or entirely non-volatile architecture, and also introduce an alternative \ac{tdc} strategy to improve system resilience in unpredictable environmental conditions. 
\end{abstract}

    \acrodef{iot}[IoT]{Internet of Things}
    \acrodef{nvm}[NVM]{Non-Volatile Memory} 
    \acrodef{aoi}[AoI]{Age of Information}
    \acrodef{hpc}[HPC]{High Performance Computing}
    \acrodef{tac}[TAC]{Time-Aware Frequency-Dependent Checkponting}
    \acrodef{pac}[PAC]{Processing-Aware Frequency-Dependent Checkpointing}
    \acrodef{vm}[VM]{Volatile Memory}
    \acrodef{ipd}[IPD]{Intermittently-Powered Device}
    \acrodef{paoi}[PAoI]{Peak Age of Information}
    \acrodef{tdc}[TDC]{Time-Dependent Checkpointing}
    \acrodef{sfc}[SFC]{Split-Frequency Checkpointing}  

\section{Introduction}
\label{sec:Introduction}

With an increasing paradigm shift towards battery-free energy harvesting-based design in low-powered embedded systems, new methods of operation have been developed to address the inherent intermittency of available harvested energy. Current \emph{intermittent computing techniques}~\cite{lucia:survey:snapl:2017,ganesan:next:hpca:2019,surbatovich:formal_intermittent:pacmpl:2020,kortbeek:legacy_software:asplos:2020,de_winkel:reliable:asplos:2020} seek to minimise the time and energy impact of power failure by strategically checkpointing---effectively saving---the system state from \ac{vm} to \ac{nvm} in mixed-memory systems~\cite{hester:batteryless:sensys:2017} (such as the popular Texas Instruments MSP430 micro-controller~\cite[Section 3.1]{lucia:survey:snapl:2017}~\cite[Section 2.1.1]{broadhead:vlc:liot:2020}). Whilst much work has been done to develop new checkpointing strategies, there is still opportunity to better describe these systems mathematically~\cite[Section 3.2.2]{broadhead:vlc:liot:2020}, in particular from a data freshness perspective. 

A relevant metric to measure data freshness is the Age of Information (\ac{aoi})~\cite{kaul:real-time_status:infocom:2012,yates:survey:arxiv:2020,kam:random_updates:isit:2013,behrouzi:freshness_leader-based_replicated:isit:2020,kaul:priority:isit:2018,baknina:status_updates:isit:2018,hsu:aoi_design:isit:2017}, which will form the basis of our analysis. We seek to model an \ac{ipd}, implementing Time-Dependent Checkpointing (TDC)~\cite[Fig.~3]{balsamo:hibernus++:ieee_transactions:2016},\cite{subasi:general_checkpoint:cluster:2017} where the \ac{vm} system state is saved to \ac{nvm} after a certain number of clock ticks have passed, and observe how fundamental system parameters (failure rate, checkpointing overhead, etc.) affect the freshness of locally sensed data. We also look to compare the mixed-memory architecture of our \ac{ipd} with single-type \ac{vm} and single-type \ac{nvm} memory structures. We go on to introduce an alternative \ac{tdc} scheme, \ac{sfc}, and consider how this can improve system resilience in unpredictable environmental conditions.

Our work builds on~\cite{ozel:timely:isit:2020} by accounting for the mixed-memory nature of many battery-free devices and the checkpointing schemes used to move data between memory types. To the best of our knowledge, this is the first evaluation and comparison of \ac{aoi} in \acp{ipd} with mixed-memory architecture. The efficiency of \ac{tdc} has been considered in other areas of research, notably for distributed stream processing~\cite{jayasekara:distributed_stream_processing:fgcs:2020} (which looked to minimise system utilization) and also in \ac{hpc} applications~\cite{di:hpc_applications:ipdps:2014} (which used \emph{wall-clock length} as the objective). However, this form of checkpointing has not been analysed for transiently-operating embedded devices, or with system freshness as the core tenet of consideration---which forms the premise of our work. 

In this paper we identify the average Peak Age of Information (\ac{paoi}) and average \ac{aoi} for an \ac{ipd} that uses \ac{tdc}. These results allow us to better understand the role of mixed-memory architecture in \acp{ipd} and how the inter-checkpointing time can be best adjusted to minimise \ac{aoi}. We also show that mixed-memory architecture can improve the system freshness of an \ac{ipd} compared with single-type \ac{vm}, however it cannot surpass entirely \ac{nvm} architecture. We further show that \ac{sfc} can improve system performance in unpredictable environmental conditions compared to inappropriately assigned single-frequency checkpoint intervals. 

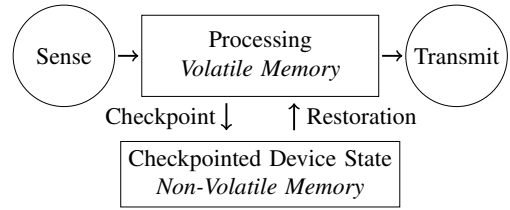
\begin{figure}
\begin{center}
\leftskip-0.4pt
\scalebox{.9}{\usetikzlibrary {shapes.geometric}
\begin{tikzpicture}[fill=white!20]
draw[help lines] (-1,-2) grid (6,3);
\path
%Sense
(-0.4,0.5) node(a) [circle,draw,fill, minimum size=15mm,inner sep=0pt] {Sense}
%Volatile memory
(2.5,0.5) node(b) [rectangle,draw,fill,align=center,minimum height=12mm,minimum width=35mm] {Processing \\ \emph{Volatile Memory}}
%Non-volatile memory
(2.5,-1.25) node(c) [rectangle,draw,fill,align=center,minimum width=35mm] {Checkpointed Device State\\ \emph{Non-Volatile Memory}}
%Transmit
(5.4,0.5) node(d) [circle,draw,fill, minimum size=15mm,inner sep=0pt] {Transmit};
%Sense to VM
\draw[thick,->] (0.4,0.5) -- (0.7,0.5);
%VM to Transmit
\draw[thick,->] (4.3,0.5) -- (4.6,0.5);
%VM to Nvm
\draw[thick,<-] (2,-0.6)--(2,-0.2); 
%NVM to VM
\draw[thick,->] (3,-0.6)--(3,-0.2); 
%Checkpoint and restoration
\coordinate (A) at (1.025,-0.422);
\node at (A){Checkpoint};
\coordinate (B) at (4,-0.39);
\node at (B){Restoration};
\end{tikzpicture}}
\vspace{-1mm}
\caption{System model for an \ac{ipd} that checkpoints its internal system state to protect from data corruption. The device performs sensing, on-board processing (that takes $P_i$ clock ticks to complete), and transmission. Sensing and transmission occur instantaneously. Data stored in \ac{vm} is checkpointed to \ac{nvm} taking $D_i$ clock ticks, and system state restoration from \ac{nvm} to \ac{vm} takes $V_i$ clock ticks. The system suffers frequent power failures.} 
\label{fig:system_diagram}
\end{center}
\vspace{-7mm}
\end{figure}

\section{System Model}\label{sec:System_model}

We consider a communication device, presented in Fig.~\ref{fig:system_diagram}, which is powered by an intermittent energy source (sun, vibrations, temperature gradient, etc.) and consequently suffers frequent power failure. The device has mixed-memory architecture (\ac{vm} and \ac{nvm}) and checkpoints the system state (processor registers, hardware registers, main memory, etc.) from \ac{vm} to \ac{nvm} after a fixed number of clock ticks, where clock ticks act as a base unit for the system's on-board clock. The sensed data is then processed and transmitted, e.g. wirelessly through a low-powered LED~\cite{broadhead:vlc:liot:2020}, to a central collecting unit.

\textbf{Packet Generation/Sensing.} Packets containing data from an on-board sensor are produced as required in a \emph{generate-at-will} type policy---where sensing only occurs when processing of the preceding packet is complete---as considered in~\cite{ozel:timely:isit:2020}. We assume that sensing is instantaneous and that this data is then immediately processed, from which a packet is created. Packets are not dropped due to power failure, since the last computation state can always be restored from \ac{nvm} to \ac{vm}.

\textbf{System Operation.} \emph{Data processing} occurs in the volatile memory of Fig.~\ref{fig:system_diagram} and encompasses a number of possible steps; such as peripheral control, filtering, and packet framing. To reduce unnecessary system complexity we have considered all processing as one stage that takes $P_i$ clock ticks to complete. The time between data sensing and packet transmission, the \emph{completion time}, is $S_i$ clock ticks and the time between two sequential packet transmissions, the \emph{inter-completion time}, is $Y_i$ clock ticks, where the idle time between a packet transmission and the next generated packet is $I_i$ clock ticks. 

\textbf{System Failure.} \acp{ipd} suffer frequent power failures due to a lack of available harvested energy. We do not explicitly consider energy as part of our system model, as in~\cite[Section V]{yates:survey:arxiv:2020},~\cite{ozel:timely:isit:2020,rafiee:energy_harvesting_node:spawc:2020}, rather assuming that the depleted energy will cause a number of random power failures. An amount of processing time $L_{i,j}$ clock ticks, where the $i$ represents the overall cycle number and $j$ the fail number within the cycle, will be wasted for each fail (since the updated system state including this processing has not been checkpointed from \ac{vm} to \ac{nvm} before failure). Our system will also be inactive for a period of $R_{i,j}$ clock ticks after failure. During power failure no new data is generated since the device cannot perform sensing. Once power is restored the system takes $V_{i,j}$ clock ticks to restore the last checkpointed system state from \ac{nvm} to \ac{vm}. For simplicity of analysis we assume that $V_{i,1} = V_{i,j}=V$ for all $j$. In practice this term would likely be fixed by design. 

\begin{figure}
\begin{center}
\scalebox{.9}{\usetikzlibrary {shapes.geometric}
\begin{tikzpicture}[fill=white!20]
\tikzstyle{every node}=[font=\small]
%%%%%%%%%%%%%%%%% Main %%%%%%%%%%%%%%%%%%%%%%%%%
%Main x-axis 
\draw[->,thick] (-3.4,2.2)--(4,2.2) node[below]{Time}; %Main y-axis
\draw[->,thick] (-4,2.2)--(-4,5.6) node[above]{AoI}; %Nodes on main x-axis from L to R
\coordinate (B) at (-3,2.2);
\node at (B){$\bigtimes$};
\coordinate (J) at (-0.3,2.19);
\node at (J){$\bullet$};
\coordinate (A) at (0,2.2);
\node at (A){$\bigtimes$};
\coordinate (H) at (1,2.19);
\node at (H){$\bullet$};
\coordinate (C) at (1.5,2.2);
\node at (C){$\bigtimes$};
\coordinate (I) at (2.5,2.19);
\node at (I){$\bullet$};
\coordinate (X) at (-0.2,3.3);
\node at (X){\textbf{$Q_{A,i}$}};
\coordinate (Y) at (2,3.215);
\node at (Y){\textbf{$Q_{A,i+1}$}};
%Main x-axis break 
\draw[-,dashed,thick] (-4,2.2)--(-3.4,2.2) node[right]{};
%Main first peak 
\draw[-,dashed] (-3,2.2)--(-0.3,4.36) node[above]{};
\draw[-,thick] (-0.3,4.36)--(-0.3,5.5) node[above]{};
\draw[-,thick] (-3.3,3.1)--(-0.3,5.5) node[above]{};
\draw[-,dashed] (-3.5,2.94)--(-3.3,3.1) node[above]{};
%\draw[-,thick] (-0.3,4.9)--(-0.3,5.9) node[above]{};
\draw[-,thick] (-0.3,4.36)--(1,5.4) node[above]{};
\draw[-,thick] (1,5.4)--(1,3) node[above]{};
%Main second peak 
\draw[-,dashed] (0,2.2)--(0.2,2.36) node[above]{};
\draw[-,dashed] (0.6,2.68)--(1,3) node[above]{};
\draw[-,thick] (1,3)--(2.5,4.2) node[above]{};
\draw[-,thick] (2.5,4.2)--(2.5,3) node[above]{};
%Main third peak 
\draw[-,dashed] (1.5,2.2)--(1.75,2.4) node[above]{};
\draw[-,dashed] (2.15,2.72)--(2.5,3) node[above]{};
\draw[-,thick] (2.5,3)--(3.5,3.8) node[above]{};
%Main dashed line box
\draw[CadetBlue,-,thick,dashed] (-3.2,1.5)--(-3.2,2.5) node[above]{};
\draw[CadetBlue,-,thick,dashed] (-3.2,1.5)--(0.2,1.5) node[above]{};
\draw[CadetBlue,-,thick,dashed] (-3.2,2.5)--(0.2,2.5) node[above]{};
\draw[CadetBlue,-,thick,dashed] (0.2,1.5)--(0.2,2.5) node[above]{};
%Main labels L to R
\draw[decoration={brace,mirror,raise=5pt},decorate]
  (-3,2.2) -- node[below=6pt] {$S_{i-1}$} (-0.3,2.2);
\draw[decoration={brace,mirror,raise=5pt},decorate]
  (-0.3,2.2) -- node[below=6pt] {$I_{i}$} (0,2.2);
\draw[decoration={brace,mirror,raise=5pt},decorate]
  (0,2.2) -- node[below=6pt] {$S_{i}$} (1,2.2);
\draw[decoration={brace,mirror,raise=5pt},decorate]
  (1,2.2) -- node[below=6pt] {$I_{i+1}$} (1.5,2.2); 
\draw[decoration={brace, raise=5pt},decorate]
  (-0.3,2.2) -- node[above=7pt] {$Y_{i}$} (1.05,2.2); 
\draw[decoration={brace, raise=5pt},decorate]
  (1,2.2) -- node[above=7pt] {$Y_{i+1}$} (2.5,2.2); 
\draw[decoration={brace,mirror,raise=5pt},decorate]
  (1.5,2.2) -- node[below=6pt] {$S_{i+1}$} (2.5,2.2);  
%Main filled Area
\begin{pgfonlayer}{bg}
\draw[thick, draw opacity =0, fill=gray!5] (0,2.2) -- (2.5,4.2) -- (2.5,3) -- (1.5,2.2) -- cycle;
\end{pgfonlayer}

%%%%%%%%% Connecting arrow %%%%%%%%%%% 
 \draw[CadetBlue,thick,->] (-1.7,1.35)--(-1.7,1.05) node[above]{};
 
%%%%%%%%%%%%%%%%%  Zoom-in %%%%%%%%%%%%%%%%%%%%%%
%Zoom-in x-axis
\draw[-,thick] (-3.8,0)--(3.4,0) node[below]{};
\draw[-,dashed,thick] (3.4,0)--(3.8,0) node[below]{};
\draw[-,dashed,thick] (-4.0,0)--(-3.8,0) node[below]{};
%Zoom-in Aoi lines
\draw[-,dashed] (-3.5,0)--(-2.5,0.8) node[below]{};
\draw[-,dashed] (2.95,0)--(3.75,0.8) node[below]{};
%Zoom-in filled Area
\begin{pgfonlayer}{bg}
\draw[thick, draw opacity =0, fill=gray!5] (2.95,0) -- (3.75,0.8) -- (3.75,0) -- cycle;
\end{pgfonlayer}
%Zoom-in dashed box 
\draw[CadetBlue,-,thick, dashed] (-4.3,0.9)--(4.3,0.9) node[above]{};
\draw[CadetBlue,-,thick, dashed] (-4.3,-1.3)--(4.3,-1.3) node[above]{};
\draw[CadetBlue,-,thick, dashed] (-4.3,-1.3)--(-4.3,0.9) node[above]{};
\draw[CadetBlue,-,thick, dashed] (4.3,-1.3)--(4.3,0.9) node[above]{};
%Nodes on zoom-in x-axis from L to R 
\coordinate (D) at (-3.5,0);
\node at (D){$\bigtimes$};
\coordinate (K) at (-2.3,0);
\node at (K){$\scriptscriptstyle\largestar$};
\coordinate (L) at (-2,0);
\node at (L){$\scriptscriptstyle\bigstar$};
%\coordinate (M) at (1.6,-0.01);
%\node at (M){$\scriptscriptstyle\bullet$};
\coordinate (G) at (-1.4,-0.01);
\node at (G){$\ddagger$};
\coordinate (Q) at (-0.2,-0.01);
\node at (Q){$\medtriangleup$};
\coordinate (R) at (0.1,-0.01);
\node at (R){$\filledmedtriangleup$};
\coordinate (N) at (1.3,-0.01);
\node at (N){$\scriptscriptstyle\largestar$};
\coordinate (F) at (1.6,-0.01);
\node at (F){$\bullet$};
\coordinate (E) at (2.95,0);
\node at (E){$\bigtimes$};
%Zoom-in labels L to R
\draw[decoration={brace,mirror,raise=5pt},decorate]
  (-3.5,0) -- node[below=6pt] {$K_{i,n}$} (-2.3,0);
\draw[decoration={brace,raise=5pt},decorate]
  (-2.3,0) -- node[above=6pt] {$D$} (-2,0);
\draw[decoration={brace,mirror,raise=5pt},decorate]
  (-2,0) -- node[below=6pt] {$L_{i,j}$} (-1.4,0);
\draw[decoration={brace,mirror,raise=23pt},decorate]
  (-3.5,0.1) -- node[below=24pt] {$S_{i-1}$} (1.6,0.1); 
\draw[decoration={brace,raise=5pt},decorate]
  (-1.4,0) -- node[above=6pt] {$R_{i,j}$} (-0.2,0);
\draw[decoration={brace,raise=5pt},decorate]
  (-0.2,0) -- node[above=6pt] {$V$} (0.1,0);
\draw[decoration={brace,mirror,raise=5pt},decorate]
  (0.1,0) -- node[below=6pt] {$K_{i,n+1}$} (1.3,0);
\draw[decoration={brace,raise=5pt},decorate]
  (1.3,0) -- node[above=6pt] {$D$} (1.6,0);
\draw[decoration={brace,mirror,raise=5pt},decorate]
  (1.6,0) -- node[below=6pt] {$I_{i}$} (2.95,0);
\end{tikzpicture}}
\vspace{-2mm}
\caption{\ac{aoi} evolution for an \ac{ipd} that checkpoints its state with a fixed regularity. Symbol notation: $\bigtimes$ denotes the start of sensing, $\ddagger$ denotes device failure, $\bullet$ denotes the end of the final checkpoint and instantaneous packet transmission, $\scriptscriptstyle\largestar$ denotes the start of checkpointing from \ac{vm} to \ac{nvm}, $\scriptscriptstyle\bigstar$ denotes the end of checkpointing from \ac{vm} to \ac{nvm}, $\medtriangleup$ denotes start of restoration from \ac{nvm} to \ac{vm}, and $\filledmedtriangleup$ denotes the end of restoration from \ac{nvm} to \ac{vm}. All variables are defined in Section~\ref{sec:System_model}.}
\label{fig:system_evolution}
\end{center}
\vspace{-7mm}
\end{figure}
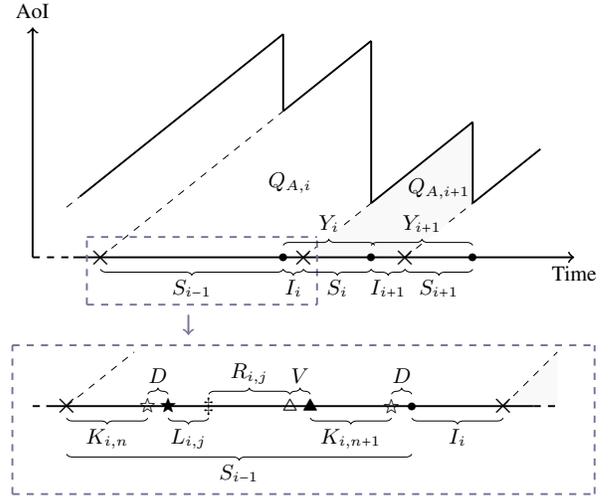

\textbf{Checkpointing Strategy.} The system will save the current device state stored in \ac{vm} to \ac{nvm} with a pre-determined regularity. This inter-checkpoint time is given by $K_{i,n}$ clock ticks where $n$ is the checkpoint number within cycle $i$. The action of checkpointing will also take a fixed amount of $D_{i,n}$ clock ticks to complete. Here, for simplicity of analysis, we assume that $D_{i,1}=D_{i,n}=D$ for all $n$. This is consistent with, for example,~\cite[Fig.~3]{balsamo:hibernus++:ieee_transactions:2016}. We also assume that $P_i$ is a multiple of the inter-checkpointing time.

\textbf{Transmission.} Our model considers the transmissions of packets to be instantaneous and occurring after a final checkpoint. 

\section{Age of Information Background}\label{sec:Age_of_Information_Background}

Let us now re-introduce several previously derived results that form the basis of our analysis. Canonically we define \ac{aoi} as~\cite[Eq. (1)]{ozel:timely:isit:2020}~\cite[Section II]{yates:survey:arxiv:2020}
\begin{equation}
\label{eq:AOI_basic}
    \Delta (t) = t - u(t), 
\end{equation}
where $\Delta (t)$ is the \ac{aoi} of data sensed by the device, $t$ is the current time, and $u(t)$ is the time stamp of the last completed packet. The \ac{aoi} time evolution for our proposed system model is depicted in Fig.~\ref{fig:system_evolution}. Using this graphical representation we can calculate the expectation of (\ref{eq:AOI_basic}) as follows. Let us denote a trapezium $Q_{A,i}$, whose area is given by
\begin{equation} \label{eq:Q_first}
    Q_{A,i} = \frac{1}{2}\left((S_{i-1}+Y_i)^2-S_i^2\right).
\end{equation} 
This is the isosceles triangle created by $S_{i-1}+Y_i$ minus the smaller isosceles triangle with base $S_i$. The area $Q_{A,i+1}$ is highlighted in Fig.~\ref{fig:system_evolution} as an example. Given that $S_{i-1}$ and $Y_i$ are independent and $S_{i-1}$ and $S_i$ are identically distributed, as considered in~\cite[Eq. (5)]{ozel:timely:isit:2020}, the expectation of (\ref{eq:Q_first}) becomes
\begin{equation} \label{eq:expectation_Q}
    \EX[Q_{A,i}] = \frac{1}{2}\EX[Y_i^2] + \EX[S_{i-1}]E[Y_i].
\end{equation}
As argued in~\cite[Eq. (6)]{ozel:timely:isit:2020}, and in concordance with the geometry presented in Fig.~\ref{fig:system_evolution}, the average \ac{aoi} is
\begin{equation} \label{eq:Delta_final}
    \EX[\Delta] = \lambda \EX[Q_{A,i}],
\end{equation}
where $\lambda$ is the interarrival time of packets to the device. Since our system model considers packet generation, rather than packet arrival, the distribution of $\lambda$ will be predicated on the distribution of time between sensing (which is equivalent to the distribution of time between transmission) hence \(\lambda = \frac{1}{E[Y]}\). Here we introduce $Y$ and $S$ for the \emph{inter-completion time} and \emph{completion time}, respectively, when taking their expected values, $\EX[Y]$ and $\EX[S]$, as time tends to infinity. This is possible due to the ergodicity of the system. As such~(\ref{eq:Delta_final}) simplifies to~\cite[Eq. (6)]{ozel:timely:isit:2020}
\begin{equation} \label{eq:Delta_simple}
    \EX[\Delta] = \frac{\EX[Y^2]}{2 \EX[Y]} + \EX[S].
\end{equation}
The average \ac{paoi}, a more manageable measure of data freshness compared to average \ac{aoi}, is given by~\cite[Eq. (7)]{ozel:timely:isit:2020}~\cite[Eq. (8)]{yates:survey:arxiv:2020}
\begin{equation} \label{eq:AoI_Peak}
    \EX[\Delta^{\text{Peak}}] = \EX[Y] + \EX[S].
\end{equation} 

\section{Completion and Inter-Completion Time} \label{sec:Completion and Inter-Completion Time}

Given~(\ref{eq:Delta_simple}) and~(\ref{eq:AoI_Peak}) it is imperative that we find expressions for $Y_i$ and $S_i$ from which we can calculate their expected values, $\EX[Y]$ and $\EX[S]$, over many cycles. From Fig.~\ref{fig:system_evolution} we see that the \emph{inter-completion time} for cycle $i$ is
\vspace{-3mm}
\begin{multline}\label{eq:Y}
    Y_i=\underbrace{\vphantom{\sum_{j=1}^{f}\left(L_{i,j}+R_{i,j}+V\right)}I_i}_{\text{Idle time}}\!+\!\underbrace{\sum_{j=1}^{f}\left(L_{i,j}+R_{i,j}+V\right)}_{\text{Events associated with failure}}\\ 
    +\underbrace{\vphantom{\sum_{j=1}^{f}\left(L_{i,j}+R_{i,j}+V\right)}\sum_{n=1}^{h}K_{i,n}}_{\text{Processing $\triangleq P_i$}}\!+\! \underbrace{\vphantom{\sum_{j=1}^{f}\left(L_{i,j}+R_{i,j}+V\right)}Dh,}_{\text{Checkpointing}}
\end{multline}
where $f$ and $h$ are the number of system fails and the number of successfully performed checkpoints during the period of time $Y_i$, respectively. The time between data being sensed and a packet being transmitted, the \emph{completion time}, for cycle $i$ is
\begin{equation}\label{eq:S}
    S_i= Y_i - I_{i}. 
\end{equation}

\section{Expectation of Completion and Inter-Completion Time} \label{sec:Expectation of Completion and Inter-Completion time}

To identify the expected values of \ac{aoi} in~(\ref{eq:Delta_simple}) and~(\ref{eq:AoI_Peak}), and hence understand the freshness of data produced by the modelled \ac{ipd}, we must first find the expected values of~(\ref{eq:Y}) and~(\ref{eq:S}) over many cycles. We assume that all variables in~(\ref{eq:Y}) and~(\ref{eq:S}) have well-defined means and known variance. Due to the assumed ergodicity of the system we can also simplify notation when taking the expected values of variables (for example, the expectation of $I_i$ is $\EX[I]$ as $t$ tends to infinity). By following the approach of~\cite[Eq. (10)]{ozel:timely:isit:2020}, we use Wald's identity to find that the expected values of \emph{inter-completion time} and \emph{completion time} are
\begin{align} 
    \EX[Y] &=\EX[I] + \EX[f]\left(\EX[L] + \EX[R] + V \right) \label{eq:exp_Y}\nonumber \\
    &\qquad+ \EX[h]\EX[K] + D\EX[h] ,\\
    \EX[S] &= \EX[Y] - \EX[I],
    \label{eq:exp_S}
\end{align}
respectively. By definition, the amount of wasted processing per failure will take a value $L_{i,j}\in[1\isep K_{i,\theta}+D]$ clock ticks where $L_{i,j}\in\mathbb{N}^+$ and $\theta$ is the checkpoint number of the processing cycle started but interrupted in $L_{i,j}$. For instance in Fig.~\ref{fig:system_evolution}, $\theta=n+1$ since the $K_{i,n+1}$ cycle was initially started in $L_{i,j}$, however could not finish before failure. We also assume that all possible values of $L_{i,j}$ are equally likely. Therefore, based on these assumptions, the expectation of wasted processing time per failure is
\begin{equation} \label{eq:exp_general}
    \EX[L] = \frac{\EX[K]+D+1}{2}.
\end{equation}
Given the definition of processing time $P_i$, given in the second summation of~(\ref{eq:Y}), its expected value will be 
\begin{equation} \label{eq:exp_P}
    \EX[P] = \EX[h]\EX[K].
\end{equation}
Substituting~(\ref{eq:exp_general}) and~(\ref{eq:exp_P}) into~(\ref{eq:exp_Y}) yields 
\begin{equation} \label{eq:exp_Y_update}
\EX[Y] =\ \EX[f]\left(C_1 +\frac{\EX[P]}{2\EX[h]}\right) +C_2 + D\EX[h],
\end{equation}
where $C_1 = \EX[R] + V + \frac{D+1}{2}$, and $C_2 = \EX[I]+ \EX[P]$ are defined for compactness of presentation.

\section{Expectation of Average \ac{paoi}}\label{sec:Expectation of peak AoI}

Given~(\ref{eq:AoI_Peak}) and expressions~(\ref{eq:exp_S}) and~(\ref{eq:exp_Y_update}) the average \ac{paoi} for our mixed-memory \ac{ipd} is therefore
\begin{multline} \label{eq:exp_delta_p_x}
\EX[\Delta^{\text{Peak}}]_{\text{MM}} = 2\EX[f]\left(C_1 +\frac{\EX[P]}{2\EX[h]}\right) +C_2 \\\qquad + \EX[P] + 2D\EX[h]. 
\end{multline}
We can see from this expression the clear conflict between over and under checkpointing. If failure occurs, which is a core phenomenon associated with an \ac{ipd}, then an increase in the checkpointing frequency (meaning an increase in $\EX[h]$) increases $\EX[\Delta^{\text{Peak}}]_{\text{MM}}$ through the final $2D\EX[h]$ term---i.e. the overhead associated with checkpointing. In contrast, an increase in $\EX[h]$ also reduces the $\frac{\EX[P]}{2\EX[h]}$ term, since more frequent checkpointing reduces the expected wasted processing time per power failure. From~(\ref{eq:exp_delta_p_x}) we also observe that the constant $C_1$ (and hence the constants $V$ and $D$) increase the average \ac{paoi}, so these should be minimised by systems designers. Further, increases in $\EX[I]$ and $\EX[R]$ will also increase $\EX[\Delta^{\text{Peak}}]_{\text{MM}}$, however we consider these to be outside the control of the designer. 

\textbf{Minimising the Average \ac{paoi}.} Due to the conflict of over and under checkpointing we can find a value of $\EX[h]$ to minimise the average \ac{paoi} as follows. Taking the derivative of~(\ref{eq:exp_delta_p_x}), with respect to $\EX[h]$ we have
\begin{equation}
    \frac{d\EX[\Delta^{\text{Peak}}]_{\text{MM}}}{d\EX[h]} = 2D - \frac{\EX[f]\EX[P]}{\EX[h]^2}. 
    \label{eq:differentiation}
\end{equation}
Then, letting (\ref{eq:differentiation}) equal to zero and solving for $\EX[h]$, we find that~(\ref{eq:exp_delta_p_x}) is minimised when
\begin{equation} \label{eq:exp_h_min_simple}
  \EX[h] = \sqrt{\frac{\EX[f]\EX[P]}{2D}}.
\end{equation}
This expression is comparable to~\cite[Eq. (7)]{di:hpc_applications:ipdps:2014} and shows that checkpoint optimisation is not inherently changed by the unique characteristics of the sensor device. Expression (\ref{eq:exp_h_min_simple}) shows that the optimum number of checkpoints in a cycle is based on three fundamental parameters of the system; $\EX[f]$, $\EX[P]$, and $\EX[D]$. This is consistent with intuition since a decrease in the expected number of failures would mean fewer required checkpoints, an increase in processing time would necessitate more checkpoints, and an increase in the checkpoint overhead would reduce its desirability. 
\section{Expectation of Average \ac{aoi}}
\label{sec:Expectation_of_average_AoI}

From~(\ref{eq:Delta_simple}) we can find an expression for the average \ac{aoi} of our system by substituting in~(\ref{eq:exp_Y_update}) and~(\ref{eq:exp_S}), and 
by finding $\EX[Y^2]$. Since we have assumed that the variance of each term in $Y$ is known, we are able to use the definition of variance to find $\EX[Y^2]$ from 
\begin{equation} \label{eq:exp_x_squared}
    \EX[Y^2] = \var(Y) + \EX[Y]^2. 
\end{equation}
By substituting~(\ref{eq:exp_x_squared}) into~(\ref{eq:Delta_simple}) the expression for the average \ac{aoi} becomes 
\begin{equation} \label{eq:exp_AoI_using_variance}
    \EX[\Delta] = \frac{\var(Y)}{2 \EX[Y]} + \frac{3\EX[Y]}{2} -\EX[I].
\end{equation}
As such, by substitution, we find that the average \ac{aoi} of the system is
\vspace{-3mm}
\begin{multline}
\EX[\Delta]_{\text{MM}} = \frac{\var(Y)}{2C_3 +\frac{C_4}{\EX[h]} + 2D\EX[h]} - \EX[I]\\
+\frac{3}{2}\left(C_3 +\frac{C_4}{2\EX[h]} + D\EX[h]\right),
\label{eq:exp_AoI_average_reduced}
\end{multline}
where $C_3 = \EX[f]C_1 + C_2$ and $C_4 = \EX[f]\EX[P]$ are defined for compactness of presentation. From~(\ref{eq:exp_AoI_average_reduced}) we see that the average \ac{aoi}, as with the average \ac{paoi}, is dependent on $\EX[f]$ and $\EX[h]$. Whilst minimisation of this expression with respect to $\EX[h]$ is not presented here, a tractable solution can be found by setting the derivative of~(\ref{eq:exp_AoI_average_reduced}) to zero and solving for $\EX[h]$. 

\begin{figure*}[!ht]
\vspace{-7mm}
\centering
\subfloat[Impact of Harvested Energy]{\includegraphics[width=0.32\textwidth]{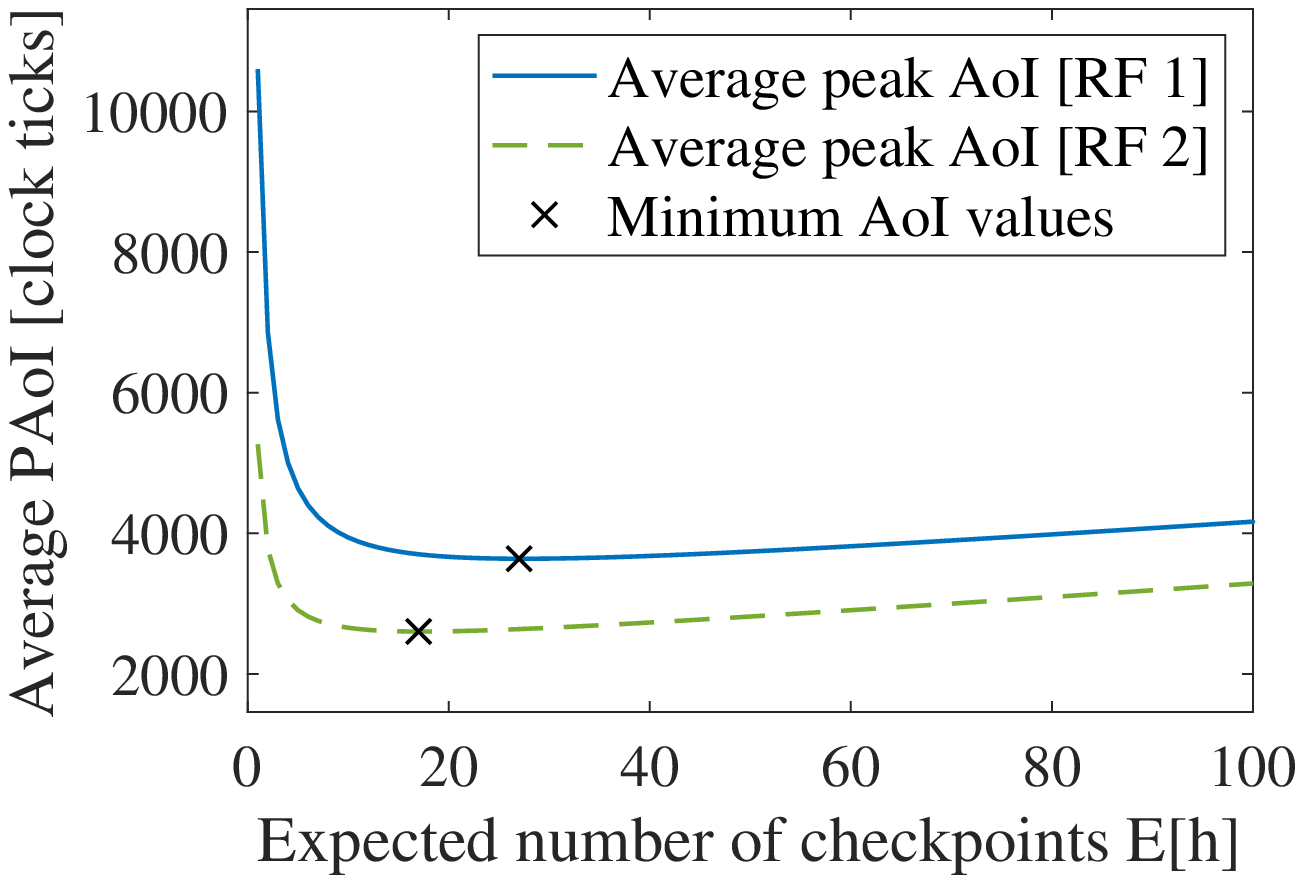}\label{fig:Aoi_peak_and_average_result}}
\hfill
\subfloat[Impact of Memory Structure]{\includegraphics[width=0.32\textwidth]{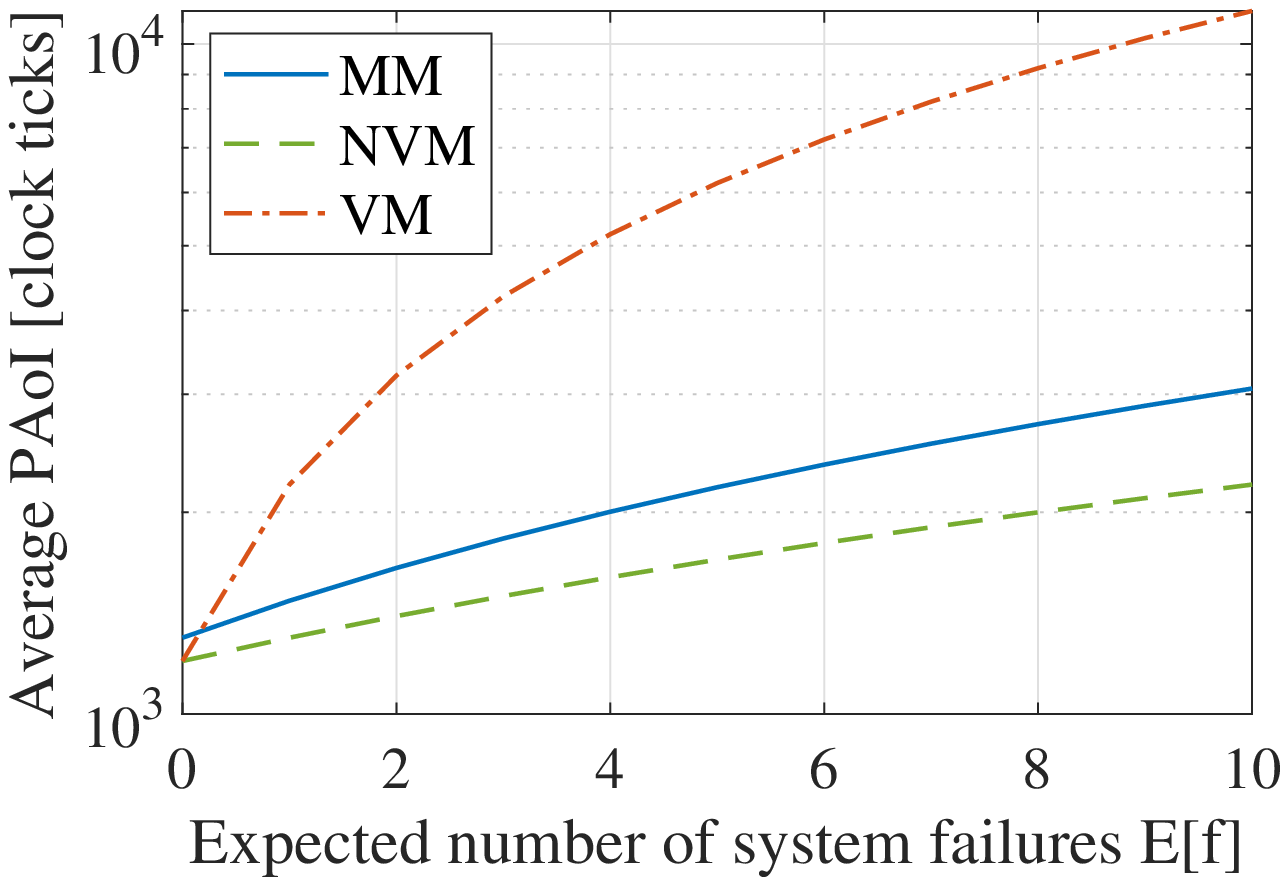}\label{fig:figure4}}
\hfill
\subfloat[Impact of Checkpointing Strategy]{\includegraphics[width=0.32\textwidth]{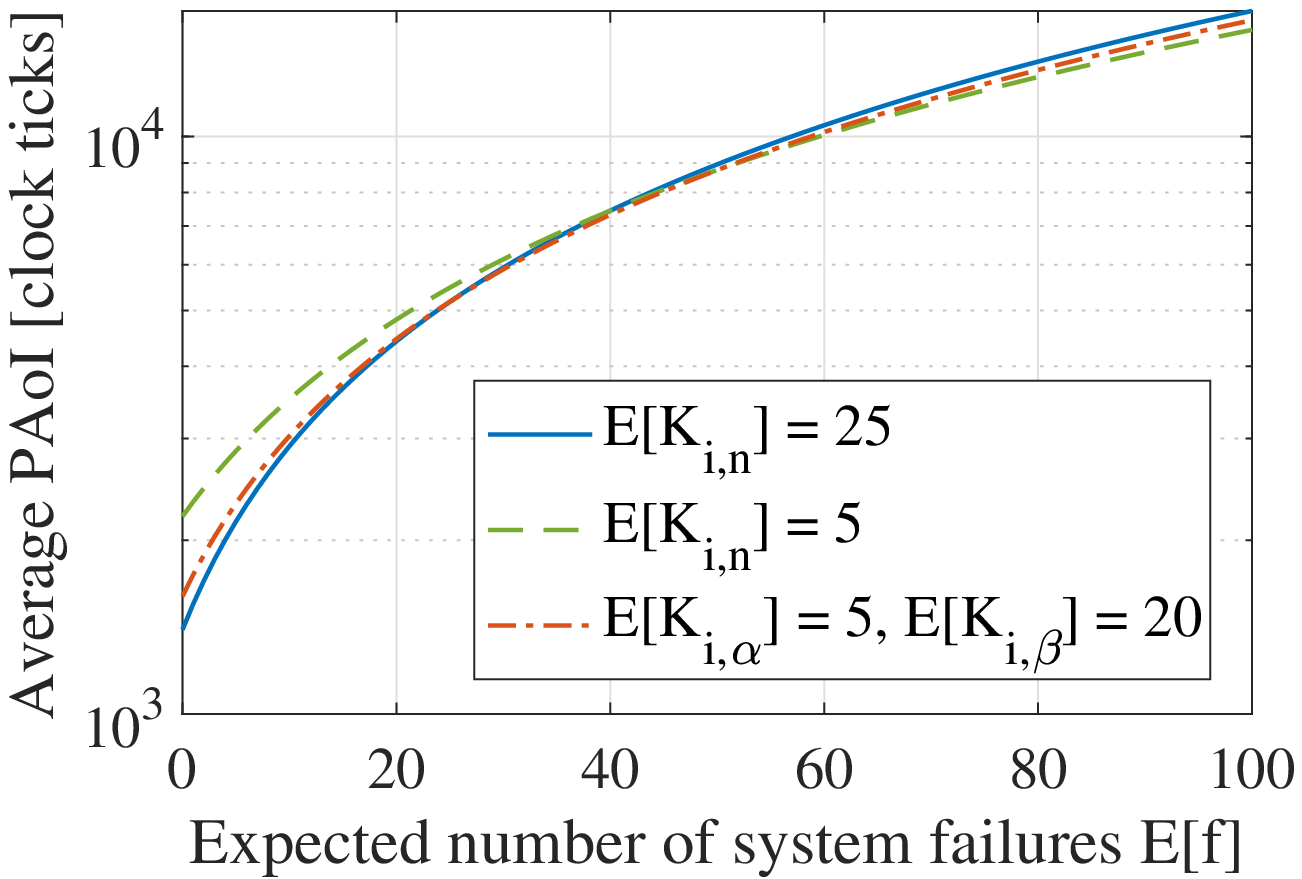}\label{fig:Split}}
\caption{Set of example numerical result (its source code is available at~\cite{broadhead:source-code:github:2021}). (a) \emph{Scenario RF 2} conditions require fewer checkpoints than \emph{Scenario RF 1}. Under-checkpointing causes significant increase in average \ac{paoi}. (b) Mixed-memory performs better than \ac{vm} for most failure conditions. (c) \ac{sfc} is better than inaccurately selected single-frequency.}
\label{fig:numerical-results}
\vspace{-4mm}
\end{figure*}

\section{Peak Age of Information: Mixed-Memory Versus Single-Memory Architecture}
\label{sec:comparison_to_no_checkpoint}

Having identified an expression for the average \ac{paoi} of our model in~(\ref{eq:exp_delta_p_x}) we proceed to examine to what extent checkpointing in mixed-memory architecture has affected the performance of our \ac{ipd} compared to a single-memory device. Whilst in practice, typical commercially available devices are mixed-memory~\cite[Section 1]{lucia:survey:snapl:2017}, we can consider two possible alternatives, a single-memory \ac{ipd} entirely comprised of \ac{nvm} \cite{su:nvp:date:2017,chien:nvm:vlsi-dat:2016}, or a single-memory \ac{ipd} entirely comprised of \ac{vm}. We assume that both types of memory are able to perform processing at the same rate.

\begin{lemma}
An entirely \ac{nvm} \ac{ipd} will have a lower or equal average \ac{paoi} than a checkpointing mixed-memory \ac{ipd}.
\end{lemma}

\begin{proof}
Using the same formulation as~(\ref{eq:Y}), the \emph{inter-completion time} for such an \ac{nvm} \ac{ipd} would be
\begin{equation}\label{eq:Y_nvm}
    Y_i = I_i + \sum_{j=1}^{f}R_{i,j} + P_i,
\end{equation}
since the device does not checkpoint and the only impact of failure is the system off-time. The \emph{completion time} would take the same form as~(\ref{eq:S}). Following the same steps as Sections~\ref{sec:Expectation of Completion and Inter-Completion time} and~\ref{sec:Expectation of peak AoI}, the average \ac{paoi} of the device would be 
\begin{equation}\label{eq:delta_nvm}
    \EX[\Delta^{\text{Peak}}]_{\text{\ac{nvm}}} = \EX[I] + 2\EX[f]\EX[R] + 2\EX[P].
\end{equation}
We can compare the above expression with the average \ac{paoi} of our system by finding the difference between~(\ref{eq:exp_delta_p_x}) and~(\ref{eq:delta_nvm}), i.e. 
\begin{multline}
    \EX[\Delta^{\text{Peak}}]_{\text{MM}}-\EX[\Delta^{\text{Peak}}]_{\text{\ac{nvm}}} = 2D\EX[h]\\ + \EX[f]\left(2V+D+1+\frac{\EX[P]}{\EX[h]}\right),
\end{multline}
which shows that
\begin{equation}
\EX[\Delta^{\text{Peak}}]_{\text{\ac{nvm}}}\leq\EX[\Delta^{\text{Peak}}]_{\text{MM}} \forall \EX[f],\EX[h],
\end{equation}
and hence a system comprised of entirely \ac{nvm} would always perform better than or equal to the mixed-memory system.
\end{proof}

\begin{lemma}
Under certain environmental conditions a mixed-memory \ac{ipd} will have a lower average \ac{paoi} than a (single-memory) VM \ac{ipd}.
\end{lemma}

\begin{proof}
A single-memory device comprised of entirely \ac{vm}, following the same formulation as~(\ref{eq:Y}), would have an \emph{inter-completion time} of
\begin{equation}\label{eq:Y_vm}
Y_i = I_i + \sum_{j=1}^{f}(R_{i,j}+\Gamma_{i,j}+I_{i,j}) + P_i, 
\end{equation}
where $\Gamma_{i,j}$ is the amount of wasted processing that occurs due to failure $j$ in cycle $i$ and the system does not checkpoint or restore (rather it re-senses after failure). $I_{i,j}$ is the idle time before the system re-senses after failure $j$. The expected value of $\Gamma_{i,j}$ will be $\EX[\Gamma]=\frac{\EX[P]+1}{2}$ since the system could waste up to $P_i$ clock ticks of processing per fail and each amount of wasted processing time is equally likely. The \emph{completion time} would take the same form as~(\ref{eq:S}). Following the same steps as Sections~\ref{sec:Expectation of Completion and Inter-Completion time} and~\ref{sec:Expectation of peak AoI}, the average \ac{paoi} of a single \ac{vm} \ac{ipd} is
\begin{multline}\label{eq:delta_vm}
    \EX[\Delta^{\text{Peak}}]_{\text{\ac{vm}}} = 2\EX[f]\left(\EX[R] + \frac{\EX[P]+1}{2} + \EX[I]\right)\\+\EX[I]+2E[P].
\end{multline}
The above expression can be greater than or smaller than~(\ref{eq:exp_delta_p_x}) depending on the selected system parameters, hence
\begin{equation}
    \exists \EX[f],\EX[h]\ni\EX[\Delta^{\text{Peak}}]_{\text{\ac{vm}}}>\EX[\Delta^{\text{Peak}}]_{\text{MM}},
\end{equation}
and thus there is a set of environments in which checkpointing in mixed-memory architecture is more efficient than not checkpointing in single-memory \ac{vm} architecture. This improvement is most evident when the mixed-memory \ac{ipd} has a low checkpointing overhead and high failure rate. 
\end{proof}

\section{Improving System Resilience in Variable Environmental Conditions}
\label{sec:improving}

Thus far we have considered a checkpointing system that can be optimised based on a known expected number of failures, yet in reality \acp{ipd} are often placed in environments with variable and unpredictable failure rates---making it difficult to pre-determine an optimum rate of checkpointing. We now propose an alternative method of \ac{tdc} for mixed-memory devices to improve system resilience---\acf{sfc}---in which the inter-checkpointing time varies between predefined intervals $\alpha$ and $\beta$. Here we once again use the framework devised in Sections~\ref{sec:System_model} and~\ref{sec:Age_of_Information_Background}. We now assume that the duration of processing between checkpoints, previously $K_{i,n}$, varies between two fixed amounts, $K_{i,\alpha}$ and $K_{i,\beta}$. Then, the processing time $P_i=\sum_{\alpha=1}^{h_\alpha}K_{i,\alpha}+\sum_{\beta=1}^{h_\beta}K_{i,\beta}$. Additionally, the expected wasted processing per fail would be $\EX[L]=p_\alpha \EX[L_{\alpha}]+p_\beta \EX[L_{\beta}]$ where $p_\alpha$ and $p_\beta$ are the probabilities of failure during an $\alpha$ and $\beta$ checkpoint, respectively, such that $p_\alpha=\frac{\EX[K_{\alpha}]}{\EX[K_{\alpha}]+\EX[K_{\beta}]}$ and $p_\beta=\frac{\EX[K_{\beta}]}{\EX[K_{\alpha}]+\EX[K_{\beta}]}$. $\EX[L_{\alpha}]$ and $\EX[L_{\beta}]$ are the expected wasted processing due to a failure in an $\alpha$ and $\beta$ checkpoint, respectively, where $\EX[L_\alpha]=\frac{\EX[K_\alpha]+D+1}{2}$ and $\EX[L_\beta]=\frac{\EX[K_\beta]+D+1}{2}$. This can also be expressed as
\begin{equation}
    \EX[L]=\frac{\EX[K_{i,\alpha}]^2+\EX[K_{i,\beta}]^2}{2(\EX[K_{i,\alpha}]+\EX[K_{i,\beta}])}+\frac{D+1}{2}.
\end{equation}
Following the same derivation of average \ac{paoi} as in Sections~\ref{sec:Expectation of Completion and Inter-Completion time} and~\ref{sec:Expectation of peak AoI}, the average \ac{paoi} of a \ac{sfc} system with two frequencies is
\vspace{-2mm}
\begin{multline}\label{eq:split}
    E[\Delta^{\text{Peak}}]_{\text{MM(split)}}= C_2 +2D(\EX[h_\alpha]+\EX[h_\beta])\\+\EX[P] +2\EX[f]\left(C_1+\frac{\EX[K_{i,\alpha}]^2+\EX[K_{i,\beta}]^2}{2([\EX[K_{i,\alpha}]+\EX[K_{i,\beta}])}\right).
\end{multline}
From this expression we observe that the average \ac{paoi} is dependent on a number of system parameters (including $\EX[P]$, $\EX[f]$, and $D$), however most interesting is the dependence on inter-checkpointing times $\EX[K_{i,\alpha}$] and $\EX[K_{i,\beta}$], which is notably different to the $\frac{\EX[P]}{2\EX[h]}\rightarrow\frac{\EX[K]}{2}$ term in~(\ref{eq:exp_delta_p_x}).

\section{Numerical Results}
\label{sec:Numerical_results}

We now provide a set of example numerical results in Fig.~\ref{fig:numerical-results} using \emph{Scenario RF 1} and \emph{Scenario RF 2} energy harvesting conditions, with data taken from~\cite[Fig. 1]{ransford:mementos:sigarch:2011} and summarized in Table~\ref{tab:data}---we note that we have converted $\si{\milli\second}$ to our base units of clock ticks therein).

\textbf{Impact of Harvested Energy.} We present the average \ac{paoi} of our considered mixed-memory \ac{ipd} system (expression~(\ref{eq:exp_delta_p_x})) in Fig.~\ref{fig:Aoi_peak_and_average_result} using the parameters of Table~\ref{tab:data}. From Fig.~\ref{fig:Aoi_peak_and_average_result} we see that the average \ac{paoi} varies under different energy conditions and that the decrease in $\EX[f]$ between \emph{RF 1} and \emph{RF 2} decreases the value of $\EX[h]$ that minimises average \ac{paoi}. We also see that under-checkpointing has a far more significant impact of data freshness than over-checkpointing.

\textbf{Impact of Memory Structure.} We also consider the relationship between memory architecture and data freshness. Fig.~\ref{fig:figure4} presents plotted expressions~(\ref{eq:exp_delta_p_x}),~(\ref{eq:delta_nvm}), and~(\ref{eq:delta_vm}) as a function of the expected number of failures $\EX[f]$ (using Table~\ref{tab:data} RF1 parameters and $\EX[h]=10$ for MM). From Fig.~\ref{fig:figure4} it is evident that, whilst not universally true, for an expected checkpointing overhead and above a low number of failures $\EX[\Delta^{\text{Peak}}]_{\text{VM}}>\EX[\Delta^{\text{Peak}}]_{\text{MM}}>\EX[\Delta^{\text{Peak}}]_{\text{NVM}}$. This shows that whilst entirely \ac{nvm} architecture will always produce the best possible average \ac{paoi}, mixed-memory structures using checkpointing can provide significant improvements in data freshness compared with entirely volatile \acp{ipd}.  

\textbf{Impact of Checkpointing Strategy.} Finally we show the impact of checkpoining strategy by plotting expressions~(\ref{eq:exp_delta_p_x}) and~(\ref{eq:split}). Results are presented in Fig.~\ref{fig:Split}. We see that the system using two inter-checkpointing times ($\EX[K_{i,\alpha}]=5$ and $\EX[K_{i,\beta}]=20$) is the most efficient for a range $30\lessapprox\EX[f]\lessapprox50$ and also provides reasonable performance for all $\EX[f]$. Whilst \ac{sfc} cannot exceed the theoretical optimum for a single frequency (expression~(\ref{eq:exp_h_min_simple})) it can provide additional resilience by reducing the risk of a very high average \ac{paoi} due to an inappropriately chosen inter-checkpoint interval in an environment with unknown or variable failure rate. 
 
\begin{table}
\begin{center}
\begin{threeparttable}
\renewcommand{\arraystretch}{1}
\caption{System Parameter Values used for Numerical Results}
\centering
\begin{tabular}{ c c c c c c c }
\toprule
 & $\EX[P]^{\diamond}$ & $\EX[R]^{\ast}$ & $\EX[f]^{\ast}$ & $\EX[I]^{\lhd}$ & $D^{\wr}$ & $V^{\dagger}$\\
\midrule
\emph{Scenario RF 1} & 500 & 50 & 15 & 200 & 5 & 10 \\
\emph{Scenario RF 2} & 500 & 75 & 6 & 200 & 5 & 10 \\
\bottomrule
\end{tabular}
\label{tab:data}
\vspace{1mm}
  \begin{minipage}{8.4cm}%
    \footnotesize $^{\diamond}$Set as a baseline for the system. This value varies significantly based on processing needs. $^{\ast}$Representative of the dynamic variation of off-time for the first two scenarios in~\cite[Fig.1]{ransford:mementos:sigarch:2011} where failure occurs approximately every $\SI{50}{\milli\second}$ and $\SI{100}{\milli\second}$, respectively. $^{\lhd}$Approximate boot time of TinyOS from~\cite[Section 2]{ransford:mementos:sigarch:2011}. $^{\wr}$Overhead can vary significantly in real-world system, $\SI{5}{\milli\second}$ is of the order of magnitude expected compared with on-time in~\cite[Fig.1]{ransford:mementos:sigarch:2011}. $^{\dagger}$Restoration overhead is typically around twice the checkpoint overhead due to additional management and fixed boot costs. 
  \end{minipage}%
\end{threeparttable}
\end{center}
\vspace{-7mm}
\end{table}

\vspace{-2mm}
\section{Conclusion}
\label{sec:Conclusion}

In this paper we have considered an \acf{ipd} with mixed-memory architecture that periodically checkpoints the system state from volatile memory to non-volatile memory---from which it can be restored should power failure occur. We have identified expressions for the average \acf{aoi} and average \acf{paoi} of the system, and found a relationship for the expected checkpointing rate that minimises the expected \ac{paoi}. We have also shown that a mixed-memory \ac{ipd} using \acf{tdc} can reduce the system \ac{paoi} compared with a single volatile memory \ac{ipd} for selected system parameters. Further, we have proposed an alternative \ac{tdc} scheme, \acl{sfc}, which can improve \ac{ipd} performance compared with inaccurately selected single-frequency checkpoint intervals. 

\IEEEtriggeratref{14}
\bibliographystyle{IEEEtran}
\bibliography{references}
\end{document}